\documentclass[a4paper,12pt]{article}
\pdfoutput=1 
\usepackage{geometry}
\usepackage[margin=10mm]{caption}
\usepackage[utf8]{inputenc}
\usepackage[english]{babel}
\usepackage{amssymb, amsmath, amsthm, amsfonts, eucal, mathalpha, mathrsfs, yfonts, lmodern, slashed, microtype}
\usepackage[sorting=none, backend=biber, style=numeric-comp, natbib=true]{biblatex}
\usepackage{array, booktabs}
\usepackage{graphicx, subcaption, tikz, tikz-cd}
\usetikzlibrary{angles, quotes, calc, shapes.geometric}
\usepackage{authblk, csquotes, hyperref}
\hypersetup{colorlinks = false , citecolor=orange, linkcolor=orange, linkbordercolor = orange, citebordercolor = orange, urlbordercolor = orange}

\theoremstyle{definition}
\newtheorem{definition}{Definition}[section]

\newtheorem{lemma}{Lemma}[section]

\newtheorem{remark}{Remark}[section]

\newcommand{\beq}{\begin{equation}}
\newcommand{\eeq}{\end{equation}}
\newcommand{\bqa} {\begin{eqnarray}}
\newcommand{\eqa} {\end{eqnarray}}

\newcommand{\ups}{\upsilon}
\newcommand{\eps}{\varepsilon}

\def \l {\left(}
\def \r {\right)}
\def \lal {\langle}
\def \ral {\rangle}

\DeclareMathOperator{\Ad}{Ad}

\DeclareMathOperator{\End}{End}

\DeclareMathOperator{\Id}{Id}

\DeclareMathOperator{\sign}{sign}

\DeclareMathOperator{\Tr}{Tr}

\newcommand{\NN}{{\mathbb N}}
\newcommand{\ZZ}{{\mathbb Z}}
\newcommand{\RR}{{\mathbb R}}
\newcommand{\CCC}{{\mathbb C}}

\newcommand{\ind}{{\theta}}
\newcommand{\fparity}{{\eta}}
\newcommand{\fparityU}{{\Theta}}

\newcommand{\sym}{{\gamma}}
\newcommand{\perm}{{\rho}}
\newcommand{\permU}{{P}}
\newcommand{\stringop}{{\kappa}}
\newcommand{\hallcond}{{\sigma}}

\newcommand{\vortex}{{v}}
\newcommand{\hilb}[1]{{\mathcal #1}}
\newcommand{\cstar}[1]{{\mathfrak{#1}}}

\title{An index for invertible phases of two-dimensional quantum spin systems}
\author{Nikita Sopenko$\footnote{School of Natural Sciences, Institute for Advanced Study, 1 Einstein Drive, Princeton, NJ 08540 USA}$}
\date{\today}

\begin{document}
\maketitle

\begin{abstract}
We define an index for invertible phases of two-dimensional fermionic and bosonic quantum spin systems without any additional symmetry. Conjecturally, it provides a microscopic definition of an invariant related to the chiral central charge of the boundary modes $c_- \bmod 24$ when the effective conformal field theory description is valid. Using this index, we prove that free fermionic systems with Chern number $\nu \bmod 48 \neq 0$ are in a non-trivial invertible phase.
\end{abstract}

\section{Introduction}

The classification of phases of non-interacting fermionic systems at zero temperature is well understood by now \cite{kitaev2009periodic, schnyder2008classification}. An important open problem is to determine how this classification changes when we allow arbitrarily strong interaction. Although this problem appears generally intractable,  one can at least hope to solve it for phases with short-range entanglement, also known as invertible phases. By definition, a system is in an invertible phase if it can be stacked with another system, such that the ground state of the combined system can be adiabatically connected to an unentangled state while preserving locality. Various topological invariants of such phases have been constructed under deformations preserving some additional symmetry (see e.g. \cite{ogata2019classification, ogatabourne, ogata2021h, sopenko2021, LocalNoether, bachmann2024many}). For example, the Hall conductance provides an invariant with respect to deformations preserving internal $U(1)$-symmetry. But little is known when no symmetry is present, except for the one-dimensional case.

In two dimensions, under favorable circumstances, the edge modes on the boundary of a quantum system can be effectively described by a conformal field theory (CFT) whose structural properties heavily rely on conformal symmetry. One important piece of data is the chiral central charge $c_-$ of the Virasoro algebra that characterizes the projective action of the diffeomorphism group in the Hilbert space of the theory on a circle. It was suggested in \cite{kitaev2006anyons} that $c_-$ is a robust characteristic of the ground state of the system that can be defined even when conformal symmetry does not emerge at the boundary. An attempt was made in \cite[Appendix D]{kitaev2006anyons} to give a microscopic definition using assumptions about the thermal properties of gapped systems. However, these assumptions are hard to justify from the first principles, and a mathematical definition of the invariant of phases associated with $c_-$ is still lacking.

The goal of this note is to define an index for invertible phases of two-dimensional fermionic and bosonic quantum spin systems that provides a microscopic definition of an invariant presumably related to the chiral central charge $c_- \bmod 24$. As an application, we compute this index for systems of free fermions which are characterized by their Chern number $\nu$, and show that it distinguishes phases with different values of $\nu \bmod 48$. It allows to prove that free fermionic systems with $\nu \bmod 48 \neq 0$ can not be adiabatically connected to a trivial state even by adding arbitrarily strong interaction.

\subsection{Main idea}

\begin{figure}
\centering
\begin{tikzpicture}[scale=.5]

\filldraw[orange, ultra thick, fill = orange!5] (-5,-2) -- (-5+4,0+4-2) -- (5+4,0+4-2) -- (5,0-2);
\filldraw[orange, ultra thick, fill = orange!5] (-5,-1) -- (-5+4,0+4-1) -- (5+4,0+4-1) -- (5,0-1);
\filldraw[orange, ultra thick, fill = orange!5] (-5,0) -- (-5+4,0+4) -- (5+4,0+4) -- (5,0);

\draw[orange, ultra thick] (-5,-1) .. controls (-1,-1) .. (0,-1/2) .. controls (1,0) .. (5,0);
\filldraw[white, ultra thick] (-5,-2) .. controls (-1,-2) .. (0,-3/2) -- (0,-2);
\filldraw[white, ultra thick] (0,-2) -- (0,-1) .. controls (1,-2) .. (5,-2);
\filldraw[orange!5, ultra thick] (0.4, -1.2) -- (0,-1) -- (0,-3/2);
\draw[orange, ultra thick] (-5,-2) .. controls (-1,-2) .. (0,-3/2) .. controls (1,-1) .. (5,-1);
\draw[orange, ultra thick] (-5,0) .. controls (-1,0) .. (0,-1) .. controls (1,-2) .. (5,-2);

\filldraw [orange] (2,-1+2) circle (5pt);
\draw[rotate around={-45:(0,-1)}, orange, dashed] (1/2,-1) arc(0:180: 14pt and 85pt);

\end{tikzpicture}
\caption{A twist defect for $N=3$.}
\label{fig:twistdefect}
\end{figure}
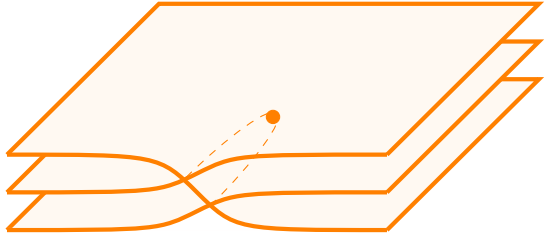

Before defining the index mathematically, let us explain the main idea behind the definition. Suppose $\Psi$ is a ground state of a two-dimensional quantum spin system with rapidly decaying correlations. We may consider a stack $\Psi^{\otimes N}$ of $N$ identical copies. It has a natural action of the symmetric group $S_N$ that permutes the systems. For a cut along a half-line, we may consider a point-like defect that corresponds to a ``flux" for the cyclic permutation $\l 1, 2, ... , N \r \to \l N, 1, 2, ... , N-1 \r$ so that the copies are permuted accordingly when one crosses the cut. Such fluxes can be interpreted as conical defects with the excess angle $2 \pi N$ or {\it twist defects} (see Fig. \ref{fig:twistdefect}). They can be made invariant under $\ZZ/N$-symmetry generated by the cyclic permutation. Under certain assumptions on the state $\Psi$, these defects are ``mobile", meaning they can be moved by applying local Hamiltonian evolution supported near a given path. If multiple defects are moved in such a way that the initial and final configurations are the same, a non-trivial phase factor may arise. We can extract a universal phase factor (topological spin) by considering the exchange process in which each path is passed in both directions (see Fig. \ref{fig:exchange}), similar to how the anyons statistics is determined for a topologically ordered system \cite{kitaev2006anyons}. This phase factor is independent of a specific Hamiltonian used to implement the movements and remains unchanged if the trajectories are deformed. We claim that it provides an invariant of the topological phase of a single state $\Psi$. 

We will show that for invertible states, the condition of mobility of twist defects is satisfied (though we expect it is satisfied, more generally, for ground states of topologically ordered systems with anyons). That allows us to define an invariant of two-dimensional invertible phases. For bosonic systems, we show that this invariant is related to the symmetry-protected index for the permutation symmetry $S_N$ of $N$ copies that takes values in the group cohomology $H^3(S_N,\RR/\ZZ)$. This group stabilizes $H^3(S_{N \geq 6}, \RR/\ZZ) = \ZZ/12 \oplus \ZZ/2 \oplus \ZZ/2$ for $N \geq 6$, providing quantization of the invariant.

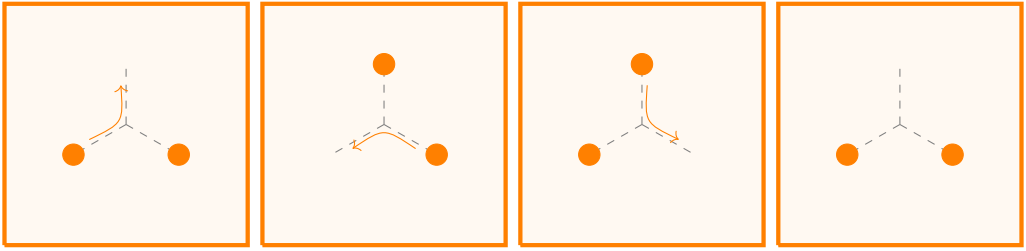
\begin{figure}
\centering
\begin{tikzpicture}[scale=.8]

\filldraw[orange, ultra thick, fill = orange!5] (-2,-2) -- (-2,2) -- (2,2) -- (2,-2) -- (-2,-2);

\draw[color=gray, dashed] (0,0) -- ({cos(0*120+90)},{sin(0*120+90)});
\draw[color=gray, dashed] (0,0) -- ({cos(1*120+90)},{sin(1*120+90)});
\draw[color=gray, dashed] (0,0) -- ({cos(2*120+90)},{sin(2*120+90)});

\filldraw [orange] ({cos(1*120+90)},{sin(1*120+90)}) circle (5pt);
\filldraw [orange] ({cos(2*120+90)},{sin(2*120+90)}) circle (5pt);

\draw[orange, <-, rotate around={-120:(0,0)}] ( {0.6*cos(1*120+90)},{0.6*sin(1*120+90)-0.1}) .. controls  (0,0-.05) .. ( {0.6*cos(2*120+90)},{0.6*sin(2*120+90)-0.1});

\end{tikzpicture}
%    \qquad % <----------------- SPACE BETWEEN PICTURES    
\begin{tikzpicture}[scale=.8]

\filldraw[orange, ultra thick, fill = orange!5] (-2,-2) -- (-2,2) -- (2,2) -- (2,-2) -- (-2,-2);

\draw[color=gray, dashed] (0,0) -- ({cos(0*120+90)},{sin(0*120+90)});
\draw[color=gray, dashed] (0,0) -- ({cos(1*120+90)},{sin(1*120+90)});
\draw[color=gray, dashed] (0,0) -- ({cos(2*120+90)},{sin(2*120+90)});

\filldraw [orange] ({cos(0*120+90)},{sin(0*120+90)}) circle (5pt);
% \filldraw [orange] ({cos(1*120+90)},{sin(1*120+90)}) circle (5pt);
\filldraw [orange] ({cos(2*120+90)},{sin(2*120+90)}) circle (5pt);

\draw[orange, <-, rotate around={0:(0,0)}] ( {0.6*cos(1*120+90)},{0.6*sin(1*120+90)-0.1}) .. controls  (0,0-.05) .. ( {0.6*cos(2*120+90)},{0.6*sin(2*120+90)-0.1});

\end{tikzpicture}
%    \qquad % <----------------- SPACE BETWEEN PICTURES    
\begin{tikzpicture}[scale=.8]

\filldraw[orange, ultra thick, fill = orange!5] (-2,-2) -- (-2,2) -- (2,2) -- (2,-2) -- (-2,-2);

\draw[color=gray, dashed] (0,0) -- ({cos(0*120+90)},{sin(0*120+90)});
\draw[color=gray, dashed] (0,0) -- ({cos(1*120+90)},{sin(1*120+90)});
\draw[color=gray, dashed] (0,0) -- ({cos(2*120+90)},{sin(2*120+90)});

\filldraw [orange] ({cos(0*120+90)},{sin(0*120+90)}) circle (5pt);
\filldraw [orange] ({cos(1*120+90)},{sin(1*120+90)}) circle (5pt);

\draw[orange, <-, rotate around={120:(0,0)}] ( {0.6*cos(1*120+90)},{0.6*sin(1*120+90)-0.1}) .. controls  (0,0-.05) .. ( {0.6*cos(2*120+90)},{0.6*sin(2*120+90)-0.1});

\end{tikzpicture}
%    \qquad % <----------------- SPACE BETWEEN PICTURES    
\begin{tikzpicture}[scale=.8]

\filldraw[orange, ultra thick, fill = orange!5] (-2,-2) -- (-2,2) -- (2,2) -- (2,-2) -- (-2,-2);

\draw[color=gray, dashed] (0,0) -- ({cos(0*120+90)},{sin(0*120+90)});
\draw[color=gray, dashed] (0,0) -- ({cos(1*120+90)},{sin(1*120+90)});
\draw[color=gray, dashed] (0,0) -- ({cos(2*120+90)},{sin(2*120+90)});

% \filldraw [orange] ({cos(0*120+90)},{sin(0*120+90)}) circle (5pt);
\filldraw [orange] ({cos(1*120+90)},{sin(1*120+90)}) circle (5pt);
\filldraw [orange] ({cos(2*120+90)},{sin(2*120+90)}) circle (5pt);

\end{tikzpicture}
\caption{An exchange process that cancels non-universal phase factor.}
\label{fig:exchange}
\end{figure}

\begin{remark} \label{rmk:cftLaughlin}
The procedure has an analog in the context of conformal field theory. One may consider $N$ copies of the same two-dimensional Euclidean conformal field theory and define twist field operators that correspond to twist defects. The topological spin of such twist fields is given by $\theta_N = e^{2 \pi i \frac{c_-}{24}(N-1/N)}$ \cite{knizhnik1987analytic}. For two-dimensional quantum many-body systems, twist defects have been discussed in \cite{barkeshli2013twist} where their statistics was analyzed using the effective field theory approach. For special Laughlin-like states of a system of particles (which are constructed using the correlation functions of a conformal field theory), the idea that the topological spin of twist defects is related to the chiral central charge has appeared in \cite{gromov2016geometric}.
\end{remark}

\begin{remark} \label{rmk:chargeambiguity}
While the procedure of producing twist defects is unambiguous at infinity, there is a local ambiguity near the point of insertion. Because of this ambiguity, the charge of the defect with respect to $\ZZ/N$ symmetry generated by the permutation is not fixed. It leads to an $N$th root of unity ambiguity in the topological spin $\theta_N$. Therefore, only $(\theta_N)^N$ gives an invariant in our construction, that fixes $c_- \bmod 24$ in conformal field theory\footnote{The same dependence on the chiral central charge also appears in the framing dependence of the Chern-Simons topological quantum field theory \cite{witten1989quantum}}. We do not know if it is possible to fix $\ZZ/N$-charge to make $\theta_N$ well-defined and provide a more refined invariant. We note that the same ambiguity appears in the choice of a $\ZZ/N$-angular momentum of an invertible state on $\RR^2$ with $\ZZ/N$ rotation symmetry. The latter is obtained from the twist defect via the map $z \to z^{1/N}$ if we regard $\RR^2$ as a complex plane with the coordinate $z \in \CCC$.
\end{remark}

For fermionic spin systems in an invertible phase, we can consider fluxes for fermionic parity. Such fluxes may also have a non-trivial phase factor in the exchange process. We will use this idea to define a $\ZZ/16$-valued index for fermionic invertible phases that distinguishes free fermionic states with different values of the Chern number $\nu \bmod 16$ in the same way as the sixteen-fold way proposed in \cite{kitaev2006anyons}. Combined with the index for permutation symmetry $S_N$ of $N$ copies, which together with the fermionic parity generates a hyperoctahedral symmetry group $B_N$, it distinguishes phases represented by free fermions with different values of the Chern number $\nu \bmod 48$.

\subsection{Organization of the paper}

In Section \ref{sec:generalities}, we give precise mathematical definitions of the systems we consider and of the equivalence classes of their states we are interested in. In Section \ref{sec:theindex}, we define the indices. We first discuss the technically easier case of bosonic systems and then explain how the discussion generalizes to fermionic systems. In Section \ref{sec:freefermions}, we consider the class of free (non-interacting) systems of fermions. We compute the indices defined in the previous section for invertible phases represented by such systems in terms of the Chern number $\nu$ and show that they depend on $\nu \bmod 48$. In Appendix \ref{app:SPTindex}, we recall the definition of a symmetry-protected index for bosonic systems.
\\

\noindent
{\bf Acknowledgements:} I'm grateful to Anton Kapustin for the discussions and comments on the draft. I would also like to thank Chris Bourne, Michael Levin, Abhinav Prem, Sahand Seifnashri, and Wilbur Shirley for useful conversations related to the subject matter of this paper. This work was supported by NSF Grant PHY-2207584 and the Ambrose Monell Foundation.
\\

\section{Generalities} \label{sec:generalities}

\subsection{Spin systems} \label{ssec:spinsystems}

Let $\hilb{K}$ be a complex Hilbert space with an anti-unitary involution $J$ (also called {\it complex conjugation} on $\hilb{K}$). There is an associated normed $*$-algebra generated by $\{c(f)\}_{f \in \hilb{K}}$ depending linearly on $f$ subject to the relations $\{c(f)^*,c(g)\} = \lal f,g \ral$\footnote{In our convention, the inner product $\lal\,\cdot\,,\,\cdot\,\ral$ on $\hilb{K}$ is anti-linear in the first argument.}, $c(f)^* = c(J f)$ with the norm $\|c(f)\| = \|f\|$. The completion of this $*$-algebra with respect to the norm is a $C^*$-algebra called {\it self-dual CAR algebra} (associated with $(\hilb{K}, J)$). Physically, this is the algebra of Majorana modes labeled by basis elements of $\hilb{K}$.

A {\it fermionic spin system} on an oriented space $\RR^d$ is defined by a finite-dimensional Hilbert space isomorphic to $\CCC^{2n}$ for some $n$ and a map $\Lambda \to \RR^d$ from a countable set $\Lambda$ called ``the lattice" that we assume to be uniformly locally finite (i.e. for some $l$, the number of points in any unit ball is less than $l$). We often identify $\Lambda$ and its subsets with their images in $\RR^d$. Elements of the lattice are called sites. Physically, the data encodes a system with $2n$ Majorana modes per site. For $\hilb{K} = l^2(\Lambda) \otimes \CCC^{2n}$ with a natural choice of complex conjugation $J$, the associated normed $*$-algebra $\cstar{A}_{\text{loc}}$ is called the algebra of local observables. It can be defined as the infinite $\ZZ/2$-graded tensor product of matrix algebras 
\beq
\cstar{A}_{\text{loc}}:= \varinjlim_{\Gamma} \cstar{A}_{\Gamma},
\eeq
\beq \label{eq:Aloc}
\cstar{A}_{\Gamma} \cong \hat{\bigotimes}_{j \in \Gamma} B(\hilb{V}_j),\footnote{Here and in the following, for a Hilbert space $\hilb{H}$, by $B(\hilb{H})$ we mean the algebra of bounded operators on $\hilb{H}$.}
\eeq
where $\hilb{V}_j \cong \hilb{V}$ are finite-dimensional $\ZZ/2$-graded Hilbert spaces isomorphic to a super vector space $\CCC^{2^{n-1}|2^{n-1}}$ and the inductive limit is taken over finite subsets $\Gamma \subset \Lambda$ with natural embeddings $\cstar{A}_{\Gamma} \to \cstar{A}_{\Gamma'}$ for $\Gamma \subset \Gamma'$. The $C^*$-algebra $\cstar{A}$ of the norm completion of $\cstar{A}_{\text{loc}}$ is called the {\it algebra of quasi-local observables}. For any $X \subset \RR^d$, we let $\cstar{A}_X \subset \cstar{A}$ be the operator norm completion of the $*$-subalgebra $\varinjlim_{\Gamma \subset X} \cstar{A}_{\Gamma}$. The algebra $\cstar{A}$ is equipped with the involutive automorphism\footnote{By an automorphism of a $*$-algebra we always mean a linear $*$-automorphism.} $\fparity$ corresponding to the grading. In the following, we use $\otimes$ for $\ZZ/2$-graded tensor products for fermionic spin systems as it is clear from the context. As usual, we call elements of $\ZZ/2$-graded vector spaces even/odd if their grading is 0/1.

The algebra of a {\it bosonic spin system} can be defined in the same way, by replacing $\hilb{V}_j$ and $B(\hilb{V}_j)$ with even subspaces $\hilb{V}^{(0)}_j \cong \CCC^{2^{n-1}|0}$ and $B(\hilb{V}^{(0)}_j)$, respectively. In that case, we can ignore the grading, and the tensor products in Eq. (\ref{eq:Aloc}) are the usual tensor products.

\begin{remark}
We could consider spin systems with arbitrary finite-dimensional on-site Hilbert spaces. However, we don't lose any generality by considering only those coming from self-dual CAR algebras (for the purposes of the classification of equivalence classes of states we are interested in) since we can always embed the former into the latter.
\end{remark}

For a compact group $G$, we say that a system {\it has $G$-symmetry} if $\hilb{V}$ is equipped with the structure of a unitary $\ZZ/2$-graded representation $u:G \to \End{\hilb{V}}$. In that case, $\cstar{A}$ is equipped with automorphisms $\{\sym^{(g)}\}_{g \in G}$ defined by $\sym^{(g)}(x) = \Ad_{u(g)} (x) := u(g)^* x u(g)$ for $x \in B(\hilb{V}_j)$, $g \in G$. These automorphisms commute with $\fparity$.

An automorphism $\alpha$ is called graded if it commutes with $\fparity$. In the presence of $G$-symmetry, we say that $\alpha$ is $G$-invariant if it commutes with $\sym^{(g)}$ for any $g$. By the abuse of notation, we often identify automorphisms $\alpha$ of $\cstar{A}_{X}$ and automorphisms $\alpha \otimes \Id_{\cstar{A}_{X^c}}$ of $\cstar{A}$ which act trivially on $\cstar{A}_{X^c}$, where $X^c$ is the complement of $X$ in $\RR^d$. We say that an automorphism $\alpha$ is {\it on-site} if there is a collection of unitary elements $\{ u_j \in B(\hilb{V}_j) \}_{j \in \Lambda}$ such that for a finite $\Gamma \subset \Lambda$ and $x \in \cstar{A}_{\Gamma}$, we have $\alpha(x) = (\prod_{j \in \Gamma} u_j)x(\prod_{j \in \Gamma} u^*_j)$. In particular, automorphisms $\fparity$ and $\sym^{(g)}$ are on-site. For any subset $X \subset \RR^d$ and an on-site automorphism $\alpha$, we let $\alpha_{X}$ be the restriction of $\alpha$ to $\cstar{A}_{X}$. 

Given two spin systems $\cstar{A}_1$, $\cstar{A}_2$, we can produce a new spin system in a natural way with the algebra of observables being $\cstar{A}_1 \otimes \cstar{A}_2$. We call this spin system a {\it stack} of two spin systems. If $\alpha$ is an automorphism of $\cstar{A}_1$, it has a natural extension $\alpha \otimes \Id$ to an automorphism of $\cstar{A}_1 \otimes \cstar{A}_2$. By the abuse of notation, we often identify $\alpha$ and its natural extension to the stack of spin systems.

\subsection{States} \label{ssec:states}

By a state over a spin system, we mean a state over the associated $C^*$-algebra. For an automorphism $\alpha$ and a state $\Psi$ over spin system $\cstar{A}$, we write $\Psi \alpha$ for the state $\Psi \alpha (x) := \Psi(\alpha(x))$, $x \in \cstar{A}$. We call a state $\Psi$ graded if $\Psi = \Psi \fparity$ and/or $G$-invariant (if the system has $G$-symmetry) if $\Psi = \Psi \sym^{(g)}$ for any $g \in G$. For a subset $X \subset \RR^d$ and a state $\Psi$ over a spin system $\cstar{A}$, by $\Psi|_X$ we mean the state of $\cstar{A}_X$ which is the restriction of the state $\Psi$. 

We say that a state over a spin system with $G$-symmetry (possibly, with trivial $G$) is {\it trivial} if it is a tensor product of vector states over $B(\hilb{V}_j)$ such that the vectors are even and $G$-invariant.

Given states $\Psi_1$, $\Psi_2$ over spin systems $\cstar{A}_1$, $\cstar{A}_2$, respectively, there is a natural state $\Psi_1 \otimes \Psi_2$ over the stack $\cstar{A}_1 \otimes \cstar{A}_2$ of $\cstar{A}_1$ and $\cstar{A}_2$ which we also call the stack (of $\Psi_1$ and $\Psi_2$). We will say that a state $\Psi$ over a spin system $\cstar{A}$ is a {\it trivial extension} of a state $\Psi_0$ over a system $\cstar{A}_0$ if $\cstar{A} \cong \cstar{A}_0 \otimes \cstar{A}_1$ and $\Psi = \Psi_0 \otimes \Psi_1$ for a trivial state $\Psi_1$ over $\cstar{A}_1$.

For a state $\Psi$ over a $C^*$-algebra $\cstar{A}$, there is an associated cyclic representation called Gelfand-Naimark-Segal (GNS) representation. We denote the Hilbert space of this representation by $\hilb{H}_{\Psi}$, the representation itself by $\pi_{\Psi}:\cstar{A} \to B(\hilb{H}_{\Psi})$ and the cyclic vector representing the state $\Psi$ by $\Omega_{\Psi} \in \hilb{H}_{\Psi}$. We have $\Psi(x) = \lal \Omega_{\Psi} , \pi_{\Psi}(x) \Omega_{\Psi} \ral$ for any $x \in \cstar{A}$. The data $(\hilb{H}_{\Psi}, \pi_{\Psi}, \Omega_{\Psi})$ is unique up to unitary equivalence. The associated von Neumann algebra $\pi_{\Psi}(\cstar{A})''$, which is the double commutant of $\pi_{\Psi}(\cstar{A})$ in $B(\hilb{H}_{\Psi})$, is denoted $\cstar{A}_{\Psi} \subset B(\hilb{H}_{\Psi})$.

Two states $\Psi_1$, $\Psi_2$ are called {\it unitarily equivalent}, if their associated representations $\pi_{\Psi_1}, \pi_{\Psi_2}$ are unitarily equivalent. For an automorphism $\alpha$ of a $C^*$-algebra $\cstar{A}$, if $\Psi$ is unitarily equivalent to $\Psi \alpha$, then there is a unitary operator $U \in B(\hilb{H}_{\Psi})$ such that $\pi_{\Psi}(\alpha(x)) = U \pi_{\Psi}(x) U^*$. In that case, we say that $\alpha$ can be {\it implemented} in the GNS representation, and call $U$ an implementer of $\alpha$. By Kadison transitivity theorem, any pure state $\Psi_1$ that is unitarily equivalent to a pure state $\Psi_2$ can be obtained from $\Psi_2$ by conjugation with a unitary element $u \in \cstar{A}$ of the $C^*$-algebra, i.e. $\Psi_1 = \Psi_2 \Ad_u$.

\subsection{Quasi-equivalence} \label{ssec:quasiequivalence}

It is useful to introduce a notion of equivalence of two states over a spin system distinguishing their behavior at infinity. A concrete way to formalize it in the context of $C^*$-algebras is the notion of {\it quasi-equivalence}. 

A state over a spin system $\cstar{A}$ is {\it factorial}\footnote{More generally, a state $\Psi$ over a $C^*$-algebra $\cstar{A}$ is called factorial if the von Neumann algebra $\cstar{A}_{\Psi}$ is a factor, i.e. its center is trivial. This more general characterization is equivalent to the one in the main text in the special case when $\cstar{A}$ is a uniformly hyperfinite algebra \cite{powers1967representations} (which is the case for spin systems).} if for any $\eps > 0$ and a finite subset $X \subset \Lambda$, there is a subset $Y \subset \Lambda$, such that $X \subset Y$ and for any $x \in \cstar{A}_X$ and $y \in \cstar{A}_{Y^c}$ we have $|\Psi(x y)-\Psi(x)\Psi(y)| \leq \eps$. Loosely speaking, this condition means that correlations between observables have decay at infinity. All pure states are factorial.

Two factorial states $\Psi_1$, $\Psi_2$ over a spin system are quasi-equivalent\footnote{More generally, two factorial states $\Psi_1$, $\Psi_2$ over a $C^*$-algebra $\cstar{A}$ are called quasi-equivalent if there is a $*$-isomorphism of von Neumann algebras $\varphi:\cstar{A}_{\Psi_1} \to \cstar{A}_{\Psi_2}$ such that $\varphi(\pi_{\Psi_1}(x)) = \pi_{\Psi_2}(x)$ for all $x \in \cstar{A}$. This more general characterization is equivalent to the one in the main text in the special case when $\cstar{A}$ is a uniformly hyperfinite algebra \cite{powers1967representations}.} if for any $\eps>0$ there is a subset $X \subset \Lambda$ such that $|\Psi_1(x) - \Psi_2(x)| \leq \eps$ for any $x \in \cstar{A}_{X^c}$. Loosely speaking, this condition means that the states coincide at infinity. If states are pure, then they are quasi-equivalent if and only if they are unitarily equivalent.

\subsection{Invertible phases} \label{ssec:SREstates}

We are interested in equivalence classes of pure states of spin systems, which are supposed to formalize mathematically the notion of a phase at zero temperature. Informally, two systems are in the same phase if their ground states can be connected by an adiabatic unitary evolution that preserves locality, possibly after adding auxiliary decoupled degrees of freedom. When the systems have gapped ground states, this definition is equivalent to saying that the Hamiltonians of the systems can be deformed into each other without closing the gap (the gap closing would correspond to phase transition). 

In the algebraic approach, unitary evolutions are implemented by automorphisms of the algebra of observables. There are different ways to characterize the locality of automorphisms mathematically, and it is not immediately clear which characterization would provide a physically reasonable definition of the phase. One framework is to take a class of Hamiltonians of the system (which are derivations of the algebra of observables) with some fixed class of decay of interactions and consider automorphism that can be generated by such Hamiltonians (i.e. unitary evolutions implemented by a one-parametric family of such Hamiltonians). For example, in \cite{LocalNoether, bachmann2024many}, the class of interactions with faster than any power decay was used to define the notion of phases and their invariants in the presence of a Lie group symmetry. The framework is motivated by the fact that ground states of gapped Hamiltonians are related by such automorphisms via spectral flow \cite{bachmann2012automorphic} provided the Hamiltonians have a sufficiently fast rate of decay. In this note, we use a weaker notion of equivalence inspired by \cite{naaijkens2022split}. The indices we define would automatically be the indices for spectral flow automorphic equivalence (see Remark \ref{rmk:autlocality}). 

\begin{figure}
\centering
\begin{tikzpicture}[scale=.5]

\filldraw[rotate around={0:(0,0)}, color=orange, fill=orange!10, very thick] ({5*cos(20)},{-5*sin(20)}) -- (0,0) -- ({5*cos(20)},{5*sin(20)});
\draw[ thick, <->, gray] ({4*cos(20)},{-4*sin(20)}) arc (-20:20:4);
\node  at (6, 0) {$2 \phi$};

\filldraw[rotate around={35:(0,0)}, color=gray, very thick] (0,0) -- (5,0);
\filldraw[rotate around={-35:(0,0)}, color=gray, very thick] (0,0) -- (5,0);
\draw[ thick, <->, gray] ({4*cos(35)},{-4*sin(35)}) arc (-35:-20:4);
\draw[ thick, <->, gray] ({4*cos(35)},{4*sin(35)}) arc (35:20:4);
\node  at ({6*cos(27.5)},{-6*sin(27.2)}) {$\theta$};
\node  at ({6*cos(27.5)},{6*sin(27.5)}) {$\theta$};

\end{tikzpicture}
\caption{$\theta$-thickening of a cone is formed by gray half-lines.}
\label{fig:thickening}
\end{figure}
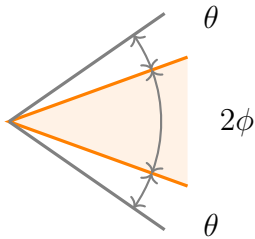

In the remainder of the paper, we consider systems on $\RR^2$. By a cone on $\RR^2$, we mean an open region associated with a point of $\RR^2$, called the apex of the cone, and an interval of a circle at infinity of $\RR^2$. The cone is spanned by the half-lines coming out of the apex and pointed in the directions of the points of the interval. By the boundary of a cone, we mean the boundary of its closure. For a cone with half-angle $\phi$ and $\theta \in (0, \pi-\phi)$, by {\it $\theta$-thickening} we mean the cone with the same apex and axis and with half-angle $\phi+\theta$ (see Fig. \ref{fig:thickening}). We will call a cone generic (for a given spin system) if it has no sites on the boundary.
By a {\it good conical partition} (see Fig. \ref{fig:ExampleOfConicalPartition}) of the lattice $\Lambda \subset \RR^2$, we mean a collection of generic cones $\{A_0, A_1,..., A_n, B_{01}, B_{12}, ..., B_{(n-1)n}, B_{n0}\}$ with the same apex such that 1) the cones $A_0,...,A_n$ are disjoint, the closure of their union is $\RR^2$, the order $0,1,2,...$ is consistent with the orientation of $\RR^2$; 2) the closures of the cones $B_{01},B_{12}, ..., B_{n0}$ intersect only at their apex; 3) the cone $B_{a(a+1)}$\footnote{In the following, for a given conical partition $\{A_0,...,A_n, B_{01},...,B_{n0}\}$, by $(a+1)$ in the subscripts we always mean $(a+1) \bmod n$} contains the half-line separating $A_a$ and $A_{a+1}$ (excluding the apex). Let $\{A_0,A_1,...,A_n,B_{01},B_{12}, ..., B_{(n-1)n}, B_{n0}\}$ be a good conical partition. By the {\it spread} of a cone $C$ with respect to this partition, we mean the union of $C$, the cones $\{ A_a \}$ that intersect $C$ and the cones $\{B_{a(a+1)} \}$ that intersect those cones $A_a$. By the {\it range} of a good conical partition we mean a minimal $\theta \in (0,\pi]$, such that for any cone $C$ with the same apex, its spread is contained in its $\theta$-thickening. 

We will say that an automorphism $\alpha$ is associated with a good conical partition if there exist automorphisms $\beta_0, \beta_1,...,\beta_n, \beta_{01},..., \beta_{n0}$ of $\cstar{A}_{A_0}$, $\cstar{A}_{A_1}$, ..., $\cstar{A}_{A_n}$, $\cstar{A}_{B_{01}}$, ... $\cstar{A}_{B_{n0}}$, respectively, such that $\alpha = \prod_{a} \beta_a \prod_{a} \beta_{a(a+1)}$. If $C$ is a generic cone with the same apex and $x \in \cstar{A}_C$, then $ \alpha(x) \in \cstar{A}_{C'}$ where $C'$ is the spread of $C$.

\begin{figure}
\centering
\begin{tikzpicture}[scale=.5]

\filldraw[rotate around={-18.4349:(0,0)}, color=orange, fill=orange!10, very thick] (5,0) -- (0,0) -- (4,3);
\filldraw[rotate around={72-18.4349:(0,0)}, color=orange, fill=orange!10, very thick] (5,0) -- (0,0) -- (4,3);
\filldraw[rotate around={2*72-18.4349:(0,0)}, color=orange, fill=orange!10, very thick] (5,0) -- (0,0) -- (4,3);
\filldraw[rotate around={3*72-18.4349:(0,0)}, color=orange, fill=orange!10, very thick] (5,0) -- (0,0) -- (4,3);
\filldraw[rotate around={4*72-18.4349:(0,0)}, color=orange, fill=orange!10, very thick] (5,0) -- (0,0) -- (4,3);

\filldraw [gray] (0,0) circle (5pt);
\filldraw[rotate around={0:(0,0)}, color=gray, very thick] (0,0) -- (5,0);
\filldraw[rotate around={1*72:(0,0)}, color=gray, very thick] (0,0) -- (5,0);
\filldraw[rotate around={2*72:(0,0)}, color=gray, very thick] (0,0) -- (5,0);
\filldraw[rotate around={3*72:(0,0)}, color=gray, very thick] (0,0) -- (5,0);
\filldraw[rotate around={4*72:(0,0)}, color=gray, very thick] (0,0) -- (5,0);

\node  at (6,0) {$B_{50}$};
\node  at ({6*cos(1*72)}, {6*sin(1*72)}) {$B_{01}$};
\node  at ({6*cos(2*72)}, {6*sin(2*72)}) {$B_{12}$};
\node  at ({6*cos(3*72)}, {6*sin(3*72)}) {$B_{23}$};
\node  at ({6*cos(4*72)}, {6*sin(4*72)}) {$B_{34}$};

\node  at ({4*cos(0*72+32)}, {4*sin(0*72+32)}) {$A_0$};
\node  at ({4*cos(1*72+32)}, {4*sin(1*72+32)}) {$A_1$};
\node  at ({4*cos(2*72+32)}, {4*sin(2*72+32)}) {$A_2$};
\node  at ({4*cos(3*72+32)}, {4*sin(3*72+32)}) {$A_3$};
\node  at ({4*cos(4*72+32)}, {4*sin(4*72+32)}) {$A_4$};

\end{tikzpicture}
\caption{A good conical partition.
}
\label{fig:ExampleOfConicalPartition}
\end{figure}
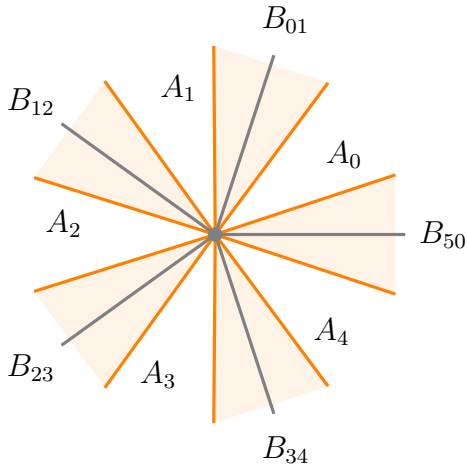

\begin{remark} \label{rmk:autlocality}
The definition of automorphisms associated with a conical partition is similar in spirit to the definition of quantum circuits. The locality, however, is imposed by a cover of a circle at infinity of $\RR^2$ rather than by the metric in the bulk. We also note that for any good conical partition, an automorphism implementing the spectral flow for a family of gapped Hamiltonians \cite{bachmann2012automorphic} can be decomposed into an inner automorphism and an automorphism associated with the partition (see \cite[Section 3]{naaijkens2022split}).
\end{remark}

\begin{definition}
We will say that two states over a spin system (possibly, with $G$-symmetry) $\Psi_0$, $\Psi_1$ are {\it in the same phase}, if for any good conical partition, there are trivial extensions $\tilde{\Psi}_0$, $\tilde{\Psi}_1$ of $\Psi_0$, $\Psi_1$ and a graded ($G$-invariant) automorphism $\alpha$ associated with the partition such that $\tilde{\Psi}_0 \alpha$ is unitarily equivalent to $\tilde{\Psi}_1$. The relation (for two states to be in the same phase) is an equivalence relation with equivalence classes called {\it phases}.   
\end{definition}
The set of phases has the structure of a commutative monoid induced by stacking. All trivial states are in the same phase which we call {\it trivial phase}, and it is a unit of the monoid.

\begin{definition}
A state $\Psi$ over a spin system (or the phase of $\Psi$) is called {\it invertible} if there is another spin system with some state $\Psi'$ over it such that $\Psi \otimes \Psi'$ is in a trivial phase. In that case, we call $\Psi'$ an inverse of $\Psi$.    
\end{definition}
Invertible phases correspond to invertible elements in the monoid and form an abelian group. By an index for invertible phases, we mean a homomorphism from this group to some abelian group.

\subsection{Split property} \label{ssec:splitproperty}

We will say that a state $\Psi$ over a spin system satisfies the {\it split property} if for any two disjoint generic cones $B_L$, $B_R$ with boundaries intersecting only at their apex, the states $\Psi|_{B_L \cup B_R}$ and $\Psi|_{B_L} \otimes  \Psi|_{B_R}$ are quasi-equivalent.

\begin{lemma} \label{lma:splitporperty}
Let $\Psi$ be an invertible state over a spin system $\cstar{A}$. Then it satisfies the split property.
\end{lemma}

\begin{proof}
Let $\Psi'$ be an inverse of $\Psi$. If the split property holds for $\Psi \otimes \Psi'$, then it holds for $\Psi$ as well. Hence, without loss of generality, we may assume that $\Psi$ in the lemma is in a trivial phase.

Let us pick a good conical partition $\{A, A^c, B_L, B_R\}$, and a generic cone $C$ with the same apex, such that $B_L$ is in $C$, while $B_R$ is in $C^c$ (see Fig. \ref{fig:conesBLBRBUBD}). We let $B_U$, $B_D$ be the cones of the complement of $B_L \cup B_R$ so that $B_U$ is contained in $A$, while $B_D$ is contained in $A^c$. For any pure state $\Psi_0$ that is quasi-equivalent to a trivial state and graded automorphisms $\beta_C$, $\beta_{C^c}$ of  $\cstar{A}_{C}$, $\cstar{A}_{C^c}$, respectively, the state $\Phi_0 = \Psi_0 \beta_C \beta_{C^c}$ is quasi-equivalent to $\Phi_0|_{C} \otimes \Phi_0|_{C^c}$. By the assumptions on $\Psi$, we can choose its trivial extension $\tilde{\Psi}$ and graded automorphisms $\beta_U$, $\beta_D$, $\beta_C$, $\beta_{C^c}$ of $\cstar{A}_{B_U}$, $\cstar{A}_{B_D}$, $\cstar{A}_{C}$, $\cstar{A}_{C^c}$, respectively, such that $\tilde{\Psi}\beta_U \beta_D \beta_C \beta_{C^c}$ is quasi-equivalent to a trivial state $\Psi_0$. Hence, 
$\Phi = \tilde{\Psi} \beta_U \beta_D$ is quasi-equivalent to $\Phi|_{C} \otimes \Phi_{C^c}$. Therefore, $\Phi|_{B_L \cup B_R}$ is quasi-equivalent to $\Phi|_{B_L} \otimes \Phi|_{B_R}$. Since the restrictions of $\Phi$ and $\tilde{\Psi}$ to $B_L \cup B_R$ coincide and $\tilde{\Psi} = \Psi \otimes \Psi'$ for some trivial state $\Psi'$, $\Psi|_{B_L \cup B_R}$ is quasi-equivalent to $\Psi|_{B_L} \otimes \Psi|_{B_R}$.
\end{proof}

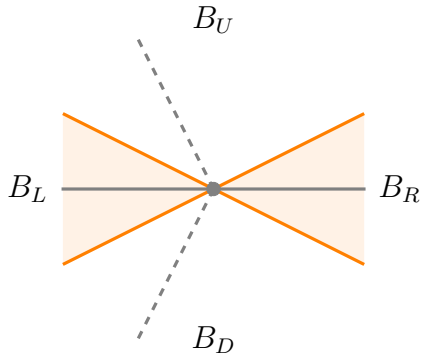
\begin{figure}
\centering
\begin{tikzpicture}[scale=.5]

\filldraw[orange,fill=orange!10, very thick] (-3.4641-1/2,-2) -- (0,0) -- (-3.4641-1/2,2);
\filldraw[orange,fill=orange!10, very thick] (3.4641+1/2,-2) -- (0,0) -- (3.4641+1/2,2);

\draw[gray, very thick, dashed] (0,0) -- (-2,-4);
\draw[gray, very thick, dashed] (0,0) -- (-2,4);

\draw[gray, very thick] (4,0) -- (0,0) -- (-4,0);

\node  at (0,-1*4) {$B_{D}$}; 
\node  at (0,1.5*3) {$B_U$}; 
\node  at (-1.5*1.7321*3/2-1,0) {$B_L$};
\node  at (1.5*1.7321*3/2+1,0) {$B_R$};

\filldraw [gray] (0,0) circle (5pt);

\end{tikzpicture}
\caption{Cones $B_L,B_R,B_U,B_D$ are separated by orange half-lines. Gray solid half-lines correspond to the boundary of $A$. Gray dashed half-lines correspond to the boundary of $C$.
}
\label{fig:conesBLBRBUBD}
\end{figure}

\begin{lemma} \label{lma:supercommutingimplementers}
Let $\Psi$ be an invertible state over a spin system $\cstar{A}$, $B_L$ and $B_R$ be disjoint cones with boundaries intersecting only at their apex, $\beta_L$ and $\beta_R$ be automorphisms of $\cstar{A}_{B_L}$ and $\cstar{A}_{B_R}$, respectively, which preserves the unitary equivalence class of $\Psi$. Let $\fparityU$, $U_L$, $U_R$ be the implementers of $\fparity$, $\beta_L$, $\beta_R$, respectively, in the GNS representation of $\Psi$. Then $U_L U_R = (-1)^{\eps_L \eps_R} U_R U_L$, where $\eps_{L,R} = 0, 1$ are defined by $\fparityU U_{L,R} = (-1)^{\eps_{L,R}} U_{L,R} \fparityU$.
\end{lemma}

\begin{proof}
We first note that if the statement of the lemma is true for a state $\Phi$, then it is also true for any state $\Psi = \Phi \alpha$ obtained from $\Phi$ by a graded automorphism $\alpha$ associated with a good conical partition with a sufficiently small range. Indeed, the phase factors in the commutants of implementers of $\fparity$, $\beta_{L,R}$ in the GNS representation of $\Psi$ and implementers of $\fparity$, $(\alpha \beta_{L,R} \alpha^{-1})$ in the GNS representation of $\Phi$ are the same, and $\alpha \beta_{L,R} \alpha^{-1}$ would be automorphisms of some thickening of $B_{L,R}$ whose closures intersect only at the apex (if the range of the good partition is sufficiently small). It is also manifest that stacking with additional systems does not affect the statement of the lemma. Since the statement holds for a state unitarily equivalent to a trivial state, it follows that it holds for any invertible state as well.
\end{proof}

\subsection{String operators} \label{ssec:stringoperators}

Suppose we have an on-site automorphism $\alpha$ of some spin system $\cstar{A}$ that leaves a state $\Psi$ invariant (i.e. $\Psi = \Psi \alpha$). Then, for a closed region $X \subset \RR^2$, the automorphism $\alpha_X$ changes $\Psi$ appreciably only near the boundary of $X$. One may wonder if it is possible to choose automorphism $\beta$ somehow localized near the boundary of $X$ such that $\Psi \alpha$ and $\Psi \beta$ are the same. Physically, if $\alpha$ is a generator of a symmetry, it would mean that the "fluxes" for this symmetry are mobile, i.e. can be moved by applying local Hamiltonian evolution supported near a given path.

We will formalize the property into the following definition.
\begin{definition}
Let $\cstar{A}$ be a spin system (possibly, with $G$-symmetry), $\Psi$ be a graded ($G$-invariant) pure state over it, and let $\alpha$ be an on-site graded ($G$-invariant) automorphism that preserves $\Psi$. We will say that $\alpha$ {\it admits string operators} for $\Psi$, if for any good conical partition $\{A, A^c, B_{L}, B_{R}\}$ there exist a trivial extension $\tilde{\Psi}$ of $\Psi$ over some (extended) spin system $\tilde{\cstar{A}}$ and graded $G$-invariant automorphisms $\beta_{L}$, $\beta_{R}$ of $\tilde{\cstar{A}}_{B_{L}}$, $\tilde{\cstar{A}}_{B_{R}}$, respectively, such that $\tilde{\Psi} \alpha_{A} \beta_{L} \beta_{R}$ is unitarily equivalent to $\tilde{\Psi}$.
\end{definition}

If $\Psi$ is invertible, then the problem of determining whether a given on-site automorphism admits string operators for $\Psi$ can be mapped to the problem of classification of invertible phases for one-dimensional systems that was essentially solved in \cite{ogatabourne} (with the definition of phases that we use in this paper). Physically, it comes from the fact that $\alpha_X$ may "pump" a non-trivial 1d invertible phase to the boundary of $X$ (when the system is in the state $\Psi$). In the remainder of this section, we show that the obstruction for fermionic spin systems with $G$-symmetry for a finite group $G$ takes values in $H^2(G, \RR/\ZZ) \times H^1(G, \ZZ/2) \times \ZZ/2$, and that it is the only obstruction to the existence of string operators. For bosonic spin systems, it reduces to an $H^2(G,\RR/\ZZ)$-valued obstruction.

Before defining this map, we first note that $\alpha$ admits string operators for $\Psi$ if and only if it admits string operators for $\Psi \otimes \Psi'$ for some inverse $\Psi'$ of $\Psi$. Indeed, if it admits string operators for $\Psi \otimes \Psi'$, then it admits string operators for $\Psi \otimes \Psi' \otimes \Psi$ and therefore also for $\Psi \otimes \Psi'_0 \otimes \Psi_0$ for some trivial states $\Psi_0$, $\Psi'_0$. Hence, it is enough to analyze the case of $\Psi$ in a trivial phase. 

Let us fix a good conical partition $\{A, A^c, B_L, B_R\}$. We let $B_U$, $B_D$ be the cones of the complement of $B_L \cup B_R$ so that $B_U$ is contained in $A$, while $B_D$ is in $A^c$ (see Fig. \ref{fig:conesBLBRBUBD}). Choose a good conical partition $\{A_0,..., A_n, B_{01}, ..., B_{n0}\}$ with the same apex as $A$ with a sufficiently small range.\footnote{If $\theta$ is less than the angles of the cones $B_L, B_U, B_R, B_D$, then the range $\theta/4$ would suffice.} We can trivially extend $\Psi$ (and by the abuse of notation we use $\Psi$ and $\cstar{A}$ for the state and the algebra of this extended system) and choose a (graded $G$-invariant) automorphism $\beta$ associated with the partition and $u \in \cstar{A}$, such that $\Psi_0 = \Psi \beta \Ad_u$ is a trivial state. Note that $\Ad_{u^*} \beta^{-1} \alpha \beta \Ad_u$ preserves $\Psi_0$. Let $\tilde{\alpha}_A = \Ad_{u^*} \beta^{-1} \alpha_A \beta \Ad_u$ and $\tilde{\Psi}_0 = \Psi_0 \tilde{\alpha}_A$. The states $\tilde{\Psi}_0|_{B_U}$, $\tilde{\Psi}_0|_{B_D}$ are quasi-equivalent to trivial states $\Psi_0|_{B_U}$, $\Psi_0|_{B_D}$, respectively, while $\tilde{\Psi}_0|_{B_L \cup B_R}$ is quasi-equivalent to $\tilde{\Psi}_0|_{B_L} \otimes \tilde{\Psi}_0|_{B_R}$. Any pure state $\Phi$, such that $\Phi|_C$ is quasi-equivalent to a trivial state for a cone $C$, is quasi-equivalent to $\Phi|_C \otimes \Phi|_{C^c}$. Hence, $\tilde{\Psi}_0$ is quasi-equivalent to $\tilde{\Psi}_0|_{B_L} \otimes \tilde{\Psi}_0|_{B_U} \otimes \tilde{\Psi}_0|_{B_R} \otimes \tilde{\Psi}_0|_{B_D}$. Note that for an automorphism $\beta_{B_{L,R}}$ of $\cstar{A}_{B_{L,R}}$, $(\beta \beta_{L,R} \beta^{-1})$ is an automorphism of $\cstar{A}_{B'_{L,R}}$ for a $\theta$-thickening $B'_{L,R}$ of $B_{L,R}$ for some $\theta$ that depends on the range of the partition (and which can be made arbitrarily small). If we managed to find (graded $G$-invariant) automorphisms $\beta_{B_L}$, $\beta_{B_R}$ of $\cstar{A}_{B_L}$, $\cstar{A}_{B_R}$, respectively, such that $(\tilde{\Psi}_0|_{B_{L,R}}) \beta_{B_{L,R}}$ are quasi-equivalent to trivial states $\Psi_0|_{B_{L,R}}$, then $\tilde{\Psi}_0 \beta_{B_L} \beta_{B_R}$ would be quasi-equivalent to $\Psi_0$, and, therefore, $\Psi \alpha_A (\beta \beta_{B_L} \beta^{-1}) (\beta \beta_{B_R} \beta^{-1})$ would be quasi-equivalent to $\Psi$. Thus, the problem of finding string operators for $\Psi$ reduces to finding automorphisms $\beta_{B_L}$, $\beta_{B_R}$.

With any pure state $\Phi$, such that $\Phi$ is quasi-equivalent to $\Phi|_{C} \otimes \Phi_{C^c}$ for a generic cone $C$, one can associate an $H^2(G, \RR/\ZZ) \times H^1(G, \ZZ/2) \times \ZZ/2$-valued index \cite{ogatabourne}. $\ZZ/2$-index is non-trivial if the states $\Phi$ and $\Phi \fparity_{C}$ are not quasi-equivalent and is trivial otherwise. The element of $H^1(G, \ZZ/2)$ corresponds to a homomorphism $G \to \ZZ/2$ that tells which unitary operators implementing automorphisms $\sym^{(g)}_{C}$ in the GNS representation of $\Phi|_{C}$ anti-commute with the implementer of the fermionic parity $\fparity_{C}$, while the element of $H^2(G, \RR/\ZZ)$ characterizes the projective representation of $G$ realized by these unitaries. The index is invariant under stacking with a trivial state and under application of graded $G$-invariant automorphisms $\alpha_C$, $\alpha_{C^c}$ of $\cstar{A}_{C}$, $\cstar{A}_{C^c}$, respectively. For bosonic systems, it follows from \cite[Theorem 1.11]{ogata2019classification} that if the index vanishes, then there is a $G$-invariant automorphism $\beta$ of $\cstar{A}_C$ such that $(\Phi|_C) \beta$ is quasi-equivalent to a trivial state over $\cstar{A}_C$. For fermionic systems, such a graded $G$-invariant $\beta$ exists, because we can use the Jordan-Wigner transformation to reduce the problem to a bosonic system with $\ZZ/2 \times G$ symmetry.

We define an $H^2(G, \RR/\ZZ) \times H^1(G, \ZZ/2) \times \ZZ/2$-valued index for $\alpha$ to be the index above for the state $\Phi = \tilde{\Psi}_0$ with $C = B_L$. It is straightforward to translate the arguments from \cite{ogatabourne} to show that the index does not depend on the choices made and has the same properties under stacking.

\section{The index}   \label{sec:theindex}

\subsection{Bosonic systems} \label{ssec:bosonicindex}

Let $\Psi$ be an invertible state over a two-dimensional bosonic spin system $\cstar{A}$. Fix $N \in \NN$ and consider a state $\Psi^{\otimes N}$ over the stack $\cstar{A}^{\otimes N}$ of $N$ copies of the spin system. There is a canonical on-site action of the symmetric group $S_N$ on $\cstar{A}^{\otimes N}$ that on elements $v_1 \otimes v_2 \otimes ... \otimes v_N \in \hilb{V}_j^{\otimes N}$ is given by $v_1 \otimes ... \otimes v_N \to v_{s(1)} \otimes ... \otimes v_{s(N)}$, $s \in S_N$. The state $\Psi^{\otimes N}$ is invariant under this $S_N$-symmetry. Let $H_N = \ZZ/N$ be the subgroup of $S_N$ of cyclic permutations with the generator $\l 1, 2, ... , N \r \to \l N, 1, 2, ... , N-1 \r$. We denote the automorphism corresponding to the generator by $\perm$. The state $\Psi^{\otimes N}$, regarded as a state over a spin system with $\ZZ/N$-symmetry, is invertible. Since $H^2(\ZZ/N, \RR/\ZZ) = 0$, the automorphism $\perm$ admits string operators for $\Psi^{\otimes N}$.

\begin{figure}
\centering
\begin{tikzpicture}[scale=.5]
\filldraw[color=orange!10, fill=orange!10, ultra thick] (0,0) -- (2,-4) -- (-2,-4) -- cycle;
\filldraw[color=orange!10, fill=orange!10, ultra thick] (0,0) -- (3.4641-1,2+3.4641/2) -- (3.4641+1,2-3.4641/2) -- cycle;
\filldraw[color=orange!10, fill=orange!10, ultra thick] (0,0) -- (-3.4641+1,2+3.4641/2) -- (-3.4641-1,2-3.4641/2) -- cycle;
\draw[orange, very thick] (0,0) -- (3.4641+1,2-3.4641/2);
\draw[orange, very thick] (0,0) -- (3.4641-1,2+3.4641/2);
\draw[orange, very thick] (0,0) -- (-3.4641-1,2-3.4641/2);
\draw[orange, very thick] (0,0) -- (-3.4641+1,2+3.4641/2);
\draw[orange, very thick] (0,0) -- (-2,-4);
\draw[orange, very thick] (0,0) -- (2,-4);

\draw[gray, very thick] (0,0) -- (3.4641,2);
\draw[gray, very thick] (0,0) -- (-3.4641,2);
\draw[gray, very thick] (0,0) -- (0,-4);

\node  at (0,-1.5*4) {$B_{12}$}; 
\node  at (1.5*1.7321*4/2,1.5*4/2) {$B_{01}$};
\node  at (-1.5*1.7321*4/2,1.5*4/2) {$B_{02}$};
\node  at (0,1.5*3) {$A_0$}; 
\node  at (-1.5*1.7321*3/2,-1.5*3/2) {$A_2$};
\node  at (1.5*1.7321*3/2,-1.5*3/2) {$A_1$};

\filldraw [gray] (0,0) circle (5pt);

\end{tikzpicture}
\caption{The cones $A_0, A_1, A_2$ are separated by gray lines, the orange lines correspond to the boundaries of the cones $B_{01}, B_{12}, B_{20}$.}
\label{fig:cones}
\end{figure}
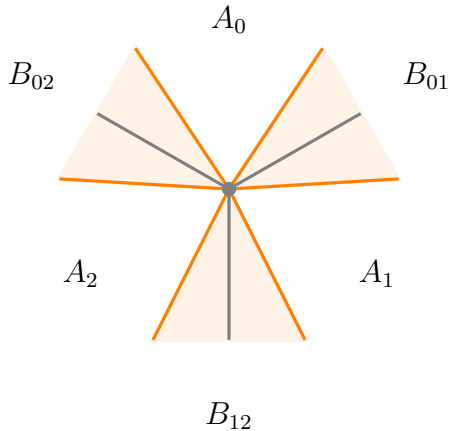

Let us choose a good conical partition $\{A_0$, $A_1$, $A_2, B_{01}, B_{12}, B_{20}\}$ (see Fig. \ref{fig:cones}). We let $\perm_{a} := \perm_{A_a}$. We choose a trivial extension of $\Psi$ (and by the abuse of notation use $\Psi$ and $\cstar{A}$ for the state and the algebra of this extended system) so that there are $\perm$-invariant automorphisms $\stringop_{a(a+1)}$ of $(\cstar{A}_{B_{a(a+1)}})^{\otimes N}$, $a \in \ZZ/3$ such that $w_a = \stringop_{(a-1)a}^{-1} \perm_{a} \stringop_{a(a+1)}$ preserve the unitary equivalence class of $\Psi^{\otimes N}$. Let $\permU$, $W_a$ be unitary operators implementing $\perm$, $w_a$, respectively, in the GNS representation of $\Psi^{\otimes N}$. Since $w_0 w_1 w_2 = \perm$, $W_0 W_1 W_2$ is proportional to $\permU$, and since $w_a$ commute with $\perm$, $\permU W_a \permU^{-1}$ is proportional to $W_a$. We can always modify $\{\stringop_{a(a+1)}\}_{a \in \ZZ/3}$ by an inner automorphism and $\{W_a\}_{a \in \ZZ/3}$ by a phase factor so that the coefficients of proportionality are $1$, and assume that such choice is made.

Since $w_a$ commute, we have 
\beq
W_a W_{a+1} = \theta_N W_{a+1} W_{a}
\eeq
for some $\theta_N \in \CCC$, $|\theta_N| = 1$ that does not depend on $a$ since $W_a$ commute with $\permU = W_0 W_1 W_2$. We claim that $\omega_N := (\theta_N)^N$ defines an invariant of the phase of $\Psi$.

To show that, we have to prove independence of $\omega_N$ of the choice of a good conical partition, automorphisms $\stringop_{a(a+1)}$, and unitary operators $W_{a}$. The operators $W_a$ only have multiplicative phase factor ambiguity, which does not affect $\theta_N$. Any two choices of $\stringop_{a(a+1)}$ differ by $\perm$-invariant automorphisms of $\cstar{A}_{B_{a(a+1)}}$ which preserve the unitary equivalence class of $\Psi^{\otimes N}$, and therefore can be implemented by unitary operators $K_{a(a+1)}$ in the GNS representation. By Lemma \ref{lma:supercommutingimplementers}, these operators commute with each other, and $K_{a(a+1)}$ commutes with $W_{a+2}$. The modified $W_a$ are given by $\tilde{W}_a = K_{(a-1)a}^{-1} W_a K_{a(a+1)}$. We have
\begin{multline}
\tilde{W}_1 \tilde{W}_0 = K_{01}^{-1} W_1 K_{12} K_{20}^{-1} W_0 K_{01} = K_{20}^{-1} K_{01}^{-1} W_1  W_0 K_{01} K_{12} = \\ = \lambda K_{20}^{-1}  W_1  W_0 K_{12} = \lambda \theta_{N}^{-1} \tilde{W}_0 \tilde{W}_1,    
\end{multline}
where $\lambda = \permU K_{01} \permU^{-1} K^{-1}_{01}$ satisfies $\lambda^N = 1$. Thus, $\omega_N = (\theta_N)^N$ is not affected. Finally, the ambiguities in the choice of the conical partition can be absorbed into the ambiguities of $\stringop_{a(a+1)}$ using the fact that any automorphism of $\cstar{A}_{X \cup Y}$ for some finite $X \subset \Lambda$ and any $Y \subset \RR^2$ is a composition of an automorphism of $\cstar{A}_Y$ and an inner automorphism of $\cstar{A}$.

Suppose two states $\Psi_1$ and $\Psi_2$ are related by an automorphism associated with a good conical partition with a sufficiently small range. Let $\beta$ be the corresponding automorphism of $\cstar{A}^{\otimes N}$ so that $\Psi_2^{\otimes N} = \Psi_1^{\otimes N} \beta$. For a given choice of a good conical partition $\{A_0$, $A_1$, $A_2, B_{01}, B_{12}, B_{20}\}$ and $w_{a}$ for the computation of $\theta_N$ for the state $\Psi_1$, we can choose a good conical partition $\{A_0$, $A_1$, $A_2, \tilde{B}_{01}, \tilde{B}_{12}, \tilde{B}_{20}\}$ with $\tilde{B}_{a(a+1)}$ being thickenings of $B_{a(a+1)}$ and $\tilde{w}_a = \beta^{-1} w_a \beta$ for the computation of $\theta_N$ for the state $\Psi_2$. The resulting $\theta_N$ for $\Psi_1$ and $\Psi_2$ would be the same, and therefore $\omega_N$ is an invariant of the phase. It is also manifest, that that for a stack of two systems, $\theta_N$ would be the product of those for each state individually, and therefore the invariant $\omega_N \in U(1)$ is a multiplicative index.

\begin{remark} \label{rmk:localcomputability}
An alternative definition of the index without the usage of the GNS representation is the following. Given that $\Psi w_a$ is unitarily equivalent to $\Psi$, by Kadison transitivity theorem, we can choose unitary elements $u_a \in \cstar{A}$, such that $\tilde{w}_a :=  w_a \Ad_{u_a}$ preserve $\Psi$. Then $\theta_N$ can be defined by $\theta_N = \Psi(w_1(u^*_1) w_0 w_1 (u^*_0 u_1) w_0(u_0))$.
\end{remark}

\begin{remark}
As mentioned in Remark \ref{rmk:chargeambiguity}, there is an ambiguity in the choice of the states for twist defects. It appears in the construction as an ambiguity in the choice of string operators $\stringop_{a(a+1)}$ which may ``pump" a non-trivial charge to the origin.
\end{remark}

\subsubsection{Relation to symmetry-protected index}

One can relate $\omega_N$ to an $S_N$-symmetry-protected index of $\Psi^{\otimes N}$ in the limit $N \to \infty$ that takes values in $H^3(BS_{\infty},\RR/\ZZ) = \ZZ/12 \oplus \ZZ/2 \oplus \ZZ/2$.

Firstly, let us relate $\omega_{N}$ to the symmetry-protected index of the state $\Psi^{\otimes N}$ for $H_N$ subgroup. Note that $w_a^{N}$ is a composition of automorphisms of $\cstar{A}_{B_{(a-1)a}}^{\otimes N}$ and $\cstar{A}_{B_{a(a+1)}}^{\otimes N}$. We let $\ups_{a(a+1)}$ be the automorphisms of $\cstar{A}_{B_{a(a+1)}}^{\otimes N}$ such that $w_a^{N} = \ups^{-1}_{(a-1)a} \ups_{a(a+1)}$. Since $\Psi$ satisfies the split property, each $\ups_{a(a+1)}$ preserves the unitary equivalence class of $\Psi$. Let $U_{a(a+1)}$ be unitary operators implementing $\ups_{a(a+1)}$ which commute with each other by Lemma \ref{lma:supercommutingimplementers}. Since $\ups_{a(a+1)}$ commute with $w_a$, we have $W_{a} U_{a(a+1)} = \omega_N U_{a(a+1)} W_{a}$, $U_{a(a+1)} W_{a+1} = \omega_N W_{a+1} U_{a(a+1)}$ for some $\omega_N \in \CCC$, $|\omega_N| = 1$ that does not depend on $a$ and $W_{a-1} U_{a(a+1)} = U_{a(a+1)} W_{a-1}$ (by Lemma \ref{lma:supercommutingimplementers}). Thus,
\beq
(\theta_N)^N = W_0^{N} W_1 W_0^{-N} W_1^{-1} = U_{20}^{-1} U_{01} W_1 U^{-1}_{01} U_{20} W_1^{-1} = \omega_N.
\eeq

Following the procedure from Appendix \ref{app:SPTindex}, we can express the cocycle for the $\ZZ/N$-symmetry protected index in terms of $\omega_N$:
\beq
\omega(a_1,a_2,a_3) = (\omega_N)^{a_1 \lfloor(a_2 + a_3)/N\rfloor},
\eeq
where $a_1,a_2,a_3 = 0, 1,...,N-1$. The symmetry-protected index of the state $\Psi^{\otimes N}$ with respect to $\ZZ/N$-symmetry is the pullback of the symmetry-protected index for $S_N$-symmetry. The third cohomology group of $S_N$ stabilizes for $N \geq 6$ and is given by $H^3(S_{N\geq 6}, \RR/\ZZ) = \ZZ/12 \oplus \ZZ/2 \oplus \ZZ/2$. By construction, the $S_{N-1}$-symmetry protected index for the state $\Psi^{\otimes(N-1)}$ is the pullback of the $S_{N}$-symmetry protected index for the state $\Psi^{\otimes N}$ with respect to the natural inclusion $S_{N-1} \to S_N$. It follows that all $\omega_N$ are determined by the class in $H^3(BS_{\infty},\RR/\ZZ)$. In particular, we have $(\omega_N)^{12} = 1$.

\begin{remark}
There is an analogous relation in the context of conformal field theory. It was noted in \cite{johnson2019moonshine} that the anomaly of the $S_N$ symmetry acting on the stack $V^{\otimes N}$ of $N$ copies of a holomorphic conformal field theory $V$ is given by $c p_1 \in H^3(BS_N,\RR/\ZZ)$, where $p_1$ is the pullback of the Pontryagin class with respect to the natural inclusion $S_N \to O(N)$ and $c$ is the central charge of the CFT. Since for a holomorphic CFT, the central charge is a multiple of $8$, and since $p_1$ has order 12 for $N > 6$, the anomaly $c p_1$ has order dividing 3.
\end{remark}

\subsection{Fermionic systems} \label{ssec:fermionicindex}

\subsubsection{\texorpdfstring{$\ZZ/16$}{Lg} index from fermionic parity}  \label{sssec:Z16}

Let $\Psi$ be an invertible state of a two-dimensional fermionic spin system. Note that the obstruction for the fermionic parity automorphism $\fparity$ to admit string operators for $\Psi$ takes values in $\ZZ/2$. 

Suppose this $\ZZ/2$ index is trivial so that $\fparity$ admits string operators for $\Psi$. We choose a good conical partition $\{A_0$, $A_1$, $A_2, B_{01}, B_{12}, B_{20}\}$ and let $\fparity_{a} := \fparity_{A_a}$. We choose a trivial extension of $\Psi$ (and by the abuse of notation use $\Psi$ and $\cstar{A}$ for the state and the algebra of this extended system) so that there are graded automorphisms $\stringop_{a(a+1)}$ of $\cstar{A}_{B_{a(a+1)}}$, $a \in \ZZ/3$ such that $w_a = \stringop_{(a-1)a}^{-1} \fparity_{a} \stringop_{a(a+1)}$ preserve the unitary equivalence class of $\Psi$. Let $\fparityU$, $W_a$ be unitary operators implementing $\fparity$, $w_a$ in the GNS representation. Since $w_0 w_1 w_2 = \fparity$, $W_0 W_1 W_2$ is proportional to $\fparityU$, and since $w_a$ commute with $\fparity$, $\fparityU W_a \fparityU^{-1} = (-1)^{\eps_a} W_a$ for $\eps_a = \pm 1$. We can always modify $\{\stringop_{a(a+1)}\}_{a=0,1,2}$ and $\{W_a\}_{a=0,1,2}$ so that the coefficients of proportionality are $1$, and assume that such choice is made.

Since $w_a$ commute, we have 
\beq
W_a W_{a+1} = \theta W_{a+1} W_{a}
\eeq
for some $\theta \in \CCC$, $|\theta_N| = 1$ that does not depend on $a$ since $W_a$ commute with $\fparityU$.

Let us analyze the ambiguities of $\theta$. Different choices of $W_a$ do not affect $\theta$. Any two choices of $\stringop_{a(a+1)}$ differ by graded automorphisms of $\cstar{A}_{B_{a(a+1)}}$ which preserve the unitary equivalence class of $\Psi$, and therefore can be implemented by unitary operators $K_{a(a+1)}$. By Lemma \ref{lma:supercommutingimplementers}, they either commute or anti-commute. Since we impose the condition that $W_a$ commute with $\fparityU$, $K_{a(a+1)}$ must all either commute with each other or anti-commute. The modified $W_a$ are given by $\tilde{W}_a = K_{(a-1)a}^{-1} W_a K_{a(a+1)}$. We have
\begin{multline}
\tilde{W}_1 \tilde{W}_0 = K_{01}^{-1} W_1 K_{12} K_{20}^{-1} W_0 K_{01} = (-1)^{\eps} K_{20}^{-1} K_{01}^{-1} W_1  W_0 K_{01} K_{12} = \\ = K_{20}^{-1}  W_1  W_0 K_{12} = \theta^{-1} \tilde{W}_0 \tilde{W}_1,
\end{multline}
where $(-1)^{\eps} = \fparityU K_{a(a+1)} \fparityU^{-1} K_{a(a+1)}^{-1} = \pm 1$. Thus, $\theta$ is not affected. Finally, as for the bosonic index defined above, ambiguities in the choice of a good conical partition can be absorbed into the ambiguities of $\stringop_{a(a+1)}$.

The same argument as in the bosonic case shows that $\theta$ is the same for states in the same phase, and provides an index for invertible states. 

Let us prove $\theta^8 = 1$. Let $\vortex_{a(a+1)} = \fparity_{a} \stringop_{a(a+1)} \fparity_{a} \stringop_{a(a+1)} = (\stringop_{a(a+1)}^{-1} \fparity_{a+1} \stringop_{a(a+1)}^{-1} \fparity_{a+1})^{-1}$ be the automorphisms of $\cstar{A}_{B_{a(a+1)}}$ which commute with each other and with $w_a$. Note that $w^2_a = v_{(a-1)a}^{-1} v_{a(a+1)}$. The automorphisms $w_a$ preserve the unitary equivalence class of $\Psi$. Since $\Psi$ satisfies the split property, the same is true for $v_{a(a+1)}$. Let us choose unitary operators $U_{a(a+1)}$ implementing $v_{a(a+1)}$ in the GNS representation. Since $v_{a(a+1)}$ are graded, we have $\fparityU U_{a(a+1)} \fparityU = (-1)^{\eps_a} U_{a(a+1)}$, $\eps_a = 0,1$. Since $W_a^2$ is proportional to $U_{(a-1)a}^{-1} U_{a(a+1)}$, all $\eps_a$ are the same $\eps_a =: \eps$. By Lemma \ref{lma:supercommutingimplementers}, the unitaries $U_{a(a+1)}$ either commute with themselves (if $\eps = 0$) or anti-commute (if $\eps = 1$). Since $v_{a(a+1)}$ and $w_{a}$ commute, $W_0 U_{01} = \omega U_{01} W_0 $ for some $\omega \in \CCC$, $|\omega| = 1$. We have $U_{01} = \omega W_0^{-1} U_{01} W_0 = \omega^2 W_0^{-2} U_{01} W^2_0 = \omega^2 (U_{20}^{-1} U_{01})^{-1} U_{01} (U_{20}^{-1} U_{01}) = (-1)^{\eps} \omega^2 U_{01}$. Thus, $\omega^2 = (-1)^{\eps}$. Finally, we note that $W^2_0 W_1 = \ind^2 W_1 W^2_0$. It implies $U_{01} W_1 = \ind^2 W_1 U_{01}$ and $\ind^2 = \omega$. Hence, we have $\theta^{8} = \omega^4 = (-1)^{2 \eps} = 1$.

Thus, we have defined a $\ZZ/2$-index for a general invertible phase and a $\ZZ/8$-index for phases with a trivial $\ZZ/2$-index. We will show in Section \ref{sec:freefermions} that for a phase represented by free fermions with Chern number $\nu$, the $\ZZ/2$-index is non-trivial when $\nu$ is odd, while $\ZZ/8$-index for even $\nu$ is given by $e^{2 \pi i \nu/16}$. It follows that $\ZZ/2$ and $\ZZ/8$ combine into a $\ZZ/16$-index for a general invertible state.

\subsubsection{Index for permutations} \label{sssec:fermionspermutation}

Let $\Psi$ be an invertible state over a two-dimensional fermionic spin system $\cstar{A}$. Fix $N \in \NN$ and consider a state $\Psi^{\otimes N}$ over the stack $\cstar{A}^{\otimes N}$ of $N$ copies of the spin system. There is a canonical on-site action of the symmetric group $S_N$ on $\cstar{A}^{\otimes N}$ that on elements $v_1 \otimes v_2 \otimes ... \otimes v_N \in \hilb{V}_j^{\otimes N}$ is given by $v_1 \otimes ... \otimes v_N \to v_{s(1)} \otimes ... \otimes v_{s(N)}$, $s \in S_N$. The state $\Psi^{\otimes N}$ is invariant under the symmetry group $B_N$ generated by permutations of the copies and fermionic parity automorphism of a single copy. $B_N$ is the hyperoctahedral group, which is the group of symmetries of an $N$-dimensional hypercube.

Let $H_N = \ZZ/N$ be the subgroup of $S_N$ of cyclic permutations with the generator $\l 1, 2, ... , N \r \to \l N, 1, 2, ... , N-1 \r$. We denote the automorphism corresponding to the generator by $\perm$. The state $\Psi^{\otimes N}$, regarded as a state over a spin system with $\ZZ/N$-symmetry, is invertible.  

Suppose $N$ is odd. The obstruction for $\perm$ to admit string operators for $\Psi ^{\otimes N}$ takes values in $H^2(\ZZ/N, \RR/\ZZ) \times H^1(\ZZ/N,\ZZ/2) \times \ZZ/2$, and since $\perm^N = \Id$ and $N$ is odd, it must be trivial. The definition of the phase factor $\theta_N$ for the exchange process of twist defects is almost identical to the definition for bosonic systems. The only difference is that the unitary operators implementing automorphisms in the construction can be odd and anti-commute with each other. We can always make $W_a$ even by modifying $\stringop_{a(a+1)}$. But because $\stringop_{a(a+1)}$ may differ by automorphism that are implemented by odd unitary operators $K_{a(a+1)}$ in the GNS representation, we have a bigger ambiguity in $\theta_N$, so that only $\omega_N = (\theta_N)^{2N}$ is unambiguous.

\begin{remark}
As for bosonic invertible phases, we expect that the index can be expressed in terms of a symmetry-protected index for $B_N$-symmetry in the limit $N \to \infty$ that should give a constraint on the order of $\omega_N$.
\end{remark}

\section{Free fermions}  \label{sec:freefermions}

In this section, we consider the class of quasi-free states of fermionic spin systems that describe ground states of quadratic (non-interacting) Hamiltonians with sufficiently rapid decay of correlations. For such states, one can define the Chern number that provides an invariant with respect to locality-preserving Bogolyubov automorphisms (i.e. automorphisms that preserve the quasi-free property). It is natural to ask if this number can be extended to an invariant of invertible (interacting) phases. We compute the index defined in the previous section for invertible quasi-free states in terms of their Chern number $\nu$ and show that states with $\nu$ that is not a multiple of 48 are in a non-trivial phase.

\subsection{Quasi-free states} \label{ssec:quasifreestates}

We recall the definition of a quasi-free state (see e.g. \cite{araki1971quasifree}). Let $\cstar{A}$ be a self-dual CAR algebra associated with a Hilbert space $\hilb{K}$ with an anti-unitary involution $J$. A state $\Phi$ over $\cstar{A}$ is called {\it quasi-free} if for any $n \in \NN$,
\beq
\Phi(c(f_1)...c(f_{2n-1})) = 0,
\eeq
\beq
\Phi(c(f_1)...c(f_{2n})) = (-1)^{n(n-1)/2} \sum_{s} \sign(s) \prod_{j=1}^{n} \Phi(c(f_{s(j)}) c(f_{s(j+n)})),
\eeq
where the sum is over permutations $s \in S_{2n}$ such that $s(1)< s(2)<...<s(n)$, $s(j)<s(j+n)$, $j = 1,...,n$. There is a bijection between quasi-free states and bounded operators $S \in B(\hilb{K})$ satisfying $0 \leq S = S^* \leq 1$, $S + J S J = 1$. The associated with $S$ quasi-free state $\Phi_S$ is defined by $\Phi_S(c(f)^* c(g)) = \lal f, S g \ral$. We call orthogonal projections $P$ such that $P + JPJ = 1$ {\it basis projections}. When $S = P$ is an orthogonal projection, the state $\Phi_P$ is pure.

With any unitary operator $U \in B(\hilb{K})$ commuting with $J$ there is an associated automorphism $\alpha_{U}$ of $\cstar{A}$ defined on generators by $\alpha_U(c(f)) = c(U f)$. We call automorphisms of this form {\it Bogolyubov automorphisms}. Note that $\Phi_P \alpha_U = \Phi_{U^* PU}$.

We can define a notion of equivalence on quasi-free states in a way similar to general states, but only allowing Bogolyubov automorphisms instead of general automorphisms. The equivalence classes with respect to this relation are called {\it free fermionic phases}. We have a natural homomorphism from the group of invertible free fermionic phases to the group of all invertible phases which is in general neither surjective nor injective.

\subsubsection{Chern number for two-dimensional fermionic systems} 

Consider a two-dimensional fermionic spin system with $\hilb{K} = l^2(\Lambda) \otimes \CCC^{2 n}$. For a subset $\Gamma \subset \Lambda$, we define an orthogonal projection $\Pi_{\Gamma} \in B(\hilb{K})$ to the closed subspace spanned by $\oplus_{j \in \Gamma} \CCC^{2n} \subset \hilb{K}$.

Let $\{A_0,A_1,A_2,B_{01},B_{12}, B_{20}\}$ be a good conical partition. Let $\Psi = \Phi_{P}$ be a quasi-free state for a basis projection $P$ that is also invertible, and let $\tilde{P} = U P U^*$ for some unitary $U \in B(\hilb{K})$. If $U$ acts as identity on $\Pi_{C^c} \hilb{K}$ for a generic cone $C$ with the same apex as $A_a$ and with a small enough angle, $[P \Pi_{A_0} P, P \Pi_{A_1} P]$ is not affected if we replace $P$ with $\tilde{P}$. If $U$ is such that $\Phi_{\tilde{P}}$ is unitarily equivalent to $\Phi_P$, then $[P,U]$ is Hilbert-Schmidt \cite{araki1971quasifree}. It implies that $[\tilde{P} \Pi_{A_0} \tilde{P}, \tilde{P} \Pi_{A_1} \tilde{P}]$ is trace-class if so is $[P \Pi_{A_0} P, P \Pi_{A_1} P]$. Thus, the operator $[P \Pi_{A_0} P, P \Pi_{A_1} P]$ is trace-class, and the Chern number
\beq
\nu(P) = 4\pi i \Tr P [P \Pi_{A_0} P, P \Pi_{A_1} P]
\eeq
for $\Phi_P$ is well-defined. It provides a quantized $\nu(P) \in \ZZ$ topological invariant of an invertible free fermionic phase \cite{bellissard1994, avron1994}.

\subsection{Computation of the index} \label{ssec:computation}

Let us consider a two-dimensional fermionic spin system with $\hilb{K} = l^2(\Lambda) \otimes \CCC^{2n}$. Let $\Psi = \Phi_P$ be a quasi-free invertible state represented by a basis projection $P \in B(\hilb{K})$. For any $N \in \NN$, let $O(N)$ act on $\hilb{K}^{\oplus N}$ by $\Id_{l^2(\Lambda)} \otimes (R^{\oplus 2n})$ for the fundamental representation $R$ of $O(N)$. The state $(\Phi_P)^{\otimes N}$ is invariant under $O(N)$-symmetry realized by Bogolyubov automorphisms corresponding to the $O(N)$ action on $\hilb{K}^{\oplus N}$. Thus, for quasi-free states, the hyperoctahedral symmetry group $B_N$ of general invertible states is enhanced to $O(N)$ symmetry.

For the remainder of the section, we fix a good conical partition $\{A_0$, $A_1$, $A_2, B_{01}, B_{12}, B_{20}\}$ (see Fig. \ref{fig:cones}). Let $q$ be an anti-symmetric self-adjoint matrix on $\CCC^{N}$ that generates $U(1)$-subgroup of $O(N)$, and let $\delta_x$ be a rank one projector in $B(l^2(\Lambda))$ that corresponds to $x \in \Lambda$. For $X \subset \RR^2$, we define operators $Q_X, \tilde{Q}_X \in B(\hilb{K})$ by $Q_X = \sum_{x \in X} \delta_x \otimes (q^{\oplus 2n})$, $\tilde{Q} = P Q_X P + (1-P) Q_X (1-P)$. Note that $[P,\tilde{Q}_X] = 0$. Let
\beq
\hallcond := 2 \pi i \Tr(P [\tilde{Q}_{A_0},\tilde{Q}_{A_1}])
\eeq
which has the physical meaning of the Hall conductance for a system with the ground state $\Psi^{\otimes N}$ and $U(1)$-symmetry realized by Bogolyubov automorphisms associated with $e^{i \alpha Q}$, $\alpha \in [0, 2 \pi)$. It is a topological invariant under deformations preserving $U(1)$ symmetry \cite{LocalNoether}. Define $V_X(\alpha) = e^{i \alpha Q_X}$ and $U_X(\alpha) = e^{i \alpha \tilde{Q}_X}$. Let $\perm_a$, $\tilde{w}_a$ be Bogolyubov automorphisms associated with $V_{A_a}(\alpha)$, $U_{A_a}(\alpha)$. We have $\tilde{w}_a = w_a \Ad_{u_a} = \stringop_{(a-1)a}^{-1} \perm_a \stringop_{a(a+1)} \Ad_{u_a}$ for some unitary elements $u_a \in \cstar{A}$ and automorphisms $\stringop_{a(a+1)}$ of $\cstar{A}_{B_{a(a+1)}}$ which commute with $\perm_a$. Then, using Remark \ref{rmk:localcomputability},
\beq
\theta = \Psi(w_1(u^*_1) w_0 w_1 (u^*_0 u_1) w_0(u_0)) = \frac{1}{2} \Tr(P U_{A_0}(\alpha_0) U_{A_1}(\alpha_1) U^{-1}_{A_0}(\alpha_0) U^{-1}_{A_1}(\alpha_1)) = e^{i \frac{\alpha_0 \alpha_1}{4 \pi} \hallcond}.
\eeq

\begin{remark}
One can interpret the automorphisms associated with $U_{X}(2 \pi /N)$ as automorphisms transporting magnetic $2 \pi/N$-fluxes along the boundary of $X$. $\theta$ computes their statistics that is given by
\beq
\theta_N = e^{\pi i \hallcond/N^2}.
\eeq
\end{remark}

Suppose $N$ is odd. Let us consider the subgroup $\ZZ/{N} \subset B_N$ generated by the permutation $\l 1, 2, ... , N \r \to \l N, 1, 2, ... , N-1 \r$. Let $q \in \End(\CCC^N)$ be
\beq
q = \sum_{j=-(N-1)/2}^{(N-1)/2} j \, e_j \otimes e^*_j,
\eeq
where $e_j$ are the eigenvectors of the $N \times N$ matrix of the permutation $\l 1, 2, ... , N \r \to \l N, 1, 2, ... , N-1 \r$ with the eigenvalue $e^{2 \pi i j/N}$. For this $q$, we have
\beq
\hallcond := \nu(P) (1^2 + 2^2 + ... +(N-1)^2/4) = \frac{\nu(P)}{24}(N^3-N),
\eeq
that gives
\beq
\theta_N = e^{\pi i \hallcond/N^2} = e^{2 \pi i \frac{\nu(P)}{48}(N-1/N)},
\eeq
\beq
(\theta_N)^{2N} = e^{2 \pi i \frac{\nu(P)}{24}(N^2-1)}.
\eeq
It follows that the invertible phases represented by quasi-free states with $\nu(P) \bmod 3 \neq 0$ are non-trivial.

Now, let us consider the action of the fermionic parity on a single copy. For quasi-free states $\Phi_P$, $\ZZ/2$-index coincides with the index defined in Appendix C.4.3 of \cite{kitaev2006anyons} and is given by $(-1)^{\nu(P)}$.

Suppose $\nu(P)$ is even so that $\ZZ/2$-index vanishes. For $X \subset \RR^2$, let $\Pi_{X} = \sum_{x \in X} \delta_x \otimes \Id_{\CCC^{2n}}$, $T = 1 - 2 P$, $\tilde{Q}_{X} = \frac12 (\Pi_X T + T \Pi_X)$ and $U_{X} (\alpha) = e^{i \alpha \tilde{Q}_{X}}$. We can choose $\tilde{w}_a$ to be the automorphisms associated with $U_{A_a}(\pi)$. Then, using
\beq
4 \pi i \Tr (P [\tilde{Q}_{A_0},\tilde{Q}_{A_1}]) = \nu{(P)},
\eeq
we get
\beq
\theta = \frac{1}{2} \Tr(P U_{A_0}(\pi) U_{A_1}(\pi) U^{-1}_{A_0}(\pi) U^{-1}_{A_1}(\pi)) = e^{2 \pi i \nu(P)/16}.
\eeq
Thus, $\ZZ/8$-index of a quasi-free invertible state with $\nu(P) = 2$ has order $8$.

Combining the statements above, we get the non-triviality of the invertible phase of a quasi-free state with $\nu(P) \bmod 48 \neq 0$.

\appendix

\section{Symmetry-protected index for bosonic systems} \label{app:SPTindex}

We will briefly review the definition of symmetry-protected indices for invertible phases of bosonic spin systems. We refer the reader to the papers where it was originally defined \cite{ogata2021h,sopenko2021} for more details. 

Let $G$ be a finite group, and let $\Psi$ be an invertible state of a bosonic spin system with $G$-symmetry. Fix a good conical partition $\{A_0, A_1, A_2, B_{01}, B_{12}, B_{20}\}$. As explained in Section \ref{ssec:stringoperators}, for any $g \in G$, $\sym^{(g)}$, regarded as an automorphism of a system with trivial symmetry, admits string operators. Let $\stringop^{(g)}_{a(a+1)}$ be automorphisms of $\cstar{A}_{B_{a(a+1)}}$, such that $w^{(g)}_0 = \stringop^{(g)-1}_{20} \sym^{(g)}_{0} \stringop^{(g)}_{01}$ preserve the unitary equivalence class of $\Psi$. We choose unitary operators $W_a^{(g)}$ implementing $w_a$ in the GNS representation.

Note that 
\beq
w^{(g_1)}_{0}w^{(g_2)}_{0}w^{(g_1 g_2)-1}_{0} = \ups^{(g_1,g_2)}_{20} \ups^{(g_1,g_2)}_{01}
\eeq 
for some automorphisms $\ups^{(g_1,g_2)}_{a(a+1)}$ of $\cstar{A}_{B_{a(a+1)}}$. Since by Lemma \ref{lma:splitporperty}, the state $\Psi$ has the split property and since $w^{(g)}_a$ preserve the unitary equivalence class of $\Psi$, $\ups^{(g_1,g_2)}_{a(a+1)}$ can be implemented by unitary operators. We choose some implementers $U^{(g_1,g_2)}_{a(a+1)}$. The unitaries $U^{(g_1,g_2)-1}_{(a-1)a} U^{(g_1,g_2)}_{a(a+1)}$ and $W^{(g_1)}_{a}W^{(g_2)}_{a}W^{(g_1 g_2)-1}_{a}$ coincide up to a phase factor. Since
\beq
\ups_{01}^{(g_1,g_2)} \ups_{01}^{(g_1 g_2,g_3)} = w_0^{(g_1)} \ups^{(g_2,g_3)}_{01} w_0^{(g_1)-1} \ups_{01}^{(g_1,g_2 g_3)},
\eeq
the operator $U_{01}^{(g_1,g_2)} U_{01}^{(g_1 g_2,g_3)} U_{01}^{(g_1,g_2 g_3)-1} (W_0^{(g_1)} U_{01}^{(g_2,g_3)} W_0^{(g_1)-1})$ is proportional to the identity operator. We define $\omega^{(g_1, g_2, g_3)} \in \CCC$, $|\omega^{(g_1, g_2, g_3)}|=1$ by
\beq
U^{(g_1,g_2)}_{01} U^{(g_1 g_2,g_3)}_{01} = \omega^{(g_1,g_2,g_3)} (W^{(g_1)}_0 U^{(g_2,g_3)}_{01} W^{(g_1)-1}_0) U^{(g_1,g_2 g_3)}_{01}.
\eeq
It is straightforward to show that it defines a cocycle for the usual cochain complex that computes the group cohomology and that the cohomology class of this cocycle $[\omega] \in H^3(G,\RR/\ZZ)$ does not depend on the choice of $A_a$, $B_{a(a+1)}$, $\stringop^{(g)}_{a(a+1)}$, $W^{(g)}_a$, $U^{(g,h)}_{a(a+1)}$ (see \cite{ogata2021h} for details). It defines an $H^3(G,\RR/\ZZ)$-valued index called symmetry-protected index.

\printbibliography

\end{document}